\documentclass[nonacm,sigconf]{acmart}
\setcopyright{none}
\renewcommand\footnotetextcopyrightpermission[1]{}
\makeatletter
\def\@ACM@checkaffil{
    \if@ACM@instpresent\else
    \ClassWarningNoLine{\@classname}{No institution present for an affiliation}%
    \fi
    \if@ACM@citypresent\else
    \ClassWarningNoLine{\@classname}{No city present for an affiliation}%
    \fi
    \if@ACM@countrypresent\else
        \ClassWarningNoLine{\@classname}{No country present for an affiliation}%
    \fi
}
\makeatother
\usepackage{color}
\usepackage{float}
\usepackage{graphicx}
\usepackage{amsfonts}

\usepackage{latexsym,amsmath}
\usepackage{todonotes}
\usepackage[ruled]{algorithm}
\usepackage[noend]{algpseudocode}
\algdef{SE}[DOWHILE]{Do}{doWhile}{\algorithmicdo}[1]{\algorithmicwhile\ #1}
\usepackage{comment}
\usepackage{epstopdf}
\usepackage{soul}
\usepackage{cleveref}
\usepackage{array}

\title{Capturing a Moving Target on Star Graphs by Two Communicating Mobile Agents}
\author{Khaled Jawhar}
\orcid{0009-0007-1074-4838}
\affiliation{
  \institution{Carleton University}
  \city{Ottawa}
  \country{Canada}
}
\email{khaledjawhar@cmail.carleton.ca}

\author{Evangelos Kranakis}
\orcid{0000-0002-8959-4428}
\affiliation{
  \institution{Carleton University}
  \city{Ottawa}
  \country{Canada}
}
\email{kranakis@scs.carleton.ca}
\begin{document}
\begin{abstract}
\noindent
We study a problem of searching for a mobile target in an $m$-ray star graph, a natural generalization of linear search to multiple directions. The target is placed adversarially on one of the rays and may move with constant speed. 
We investigate a two-robot setting, where cooperation and communication play a central role. We study two communication models: the Face-to-Face (F2F) model, where robots communicate only upon meeting, and the Sender-Receiver (S/R) model, where communication is asymmetric. We focus on the {\em away model}, in which the target moves {\em away} from the origin with speed $v<1$.

We design search strategies that minimize the competitive ratio and analyze the problem under three knowledge assumptions: \emph{NoDistance}, \emph{NoSpeed}, and \emph{NoKnowledge}. For each setting, we derive upper bounds on the competitive ratio and for some cases, we derive the lower bound. Our results highlight how the number of rays and the communication model influence the competitive ratio.

\end{abstract}

\maketitle
\begingroup
\renewcommand{\thefootnote}{}
\footnotetext{Research of Evangelos Kranakis was supported in part by NSERC (Natural Sciences and Engineering Research Council of Canada) grant.}
\endgroup
\section{Introduction}

Search and evacuation problems are fundamental topics in algorithm design, with applications in areas such as robotics, surveillance, data mining, and network exploration. In the classical linear search problem, one or more agents must locate an unknown target on an infinite line, while in evacuation problems multiple agents cooperate so that they all find an unknown exit. These problems have been extensively studied under various assumptions on the number of agents, their speeds, and the amount of information available.

A central measure of performance is the \emph{competitive ratio}, which compares the time required by an online search strategy to find the target without knowledge of its location to the time required by an optimal offline strategy that knows the target’s location in advance. An algorithm is optimal if it minimizes this ratio.

While most existing work focuses on static targets and simple domains such as the line, less attention has been given to more complex environments and dynamic targets. In this paper, we study the problem of searching for a mobile target in an $m$-ray star graph, a natural generalization of the linear domain to multiple directions. The presence of multiple rays introduces a fundamental trade-off between exploring new directions and revisiting previously explored regions, while the motion of the target further increases the difficulty.

We investigate two-robot settings under two communication models: the Face-to-Face (F2F) model, where robots communicate only upon meeting, and the Sender-Receiver (S/R) model, which allows asymmetric wireless communication. We analyze the problem under different knowledge assumptions about the target, namely whether its initial distance, speed, or neither is known.
The three knowledge assumptions are chosen to capture different levels of information available to the searching robots and to investigate how such information affects the difficulty of the capture problem. In particular, they provide a natural progression from settings in which the robots possess useful information about the target to the completely uninformed setting. Studying these variants allows us to quantify the benefit of prior knowledge through the resulting competitive ratios and to identify how the availability of information interacts with the topology of the underlying graph. These assumptions are also motivated by the corresponding settings considered for the infinite line in previous work \cite{coleman2023line,escape2024}, allowing us to examine how the transition from a line to an $m$-ray star graph affects the role of information in online target capture.
Our goal is to design search strategies that minimize the competitive ratio of capture time (also referred to as evacuation time) and to understand how the interplay between domain structure, target dynamics, communication constraints, and available knowledge affects performance. We provide upper bounds on the competitive ratio for each setting and establish lower bounds for other cases.
\subsection{Problem Formulation and Objectives}
The primary objective of this work is to design online search strategies for two mobile agents that minimize the \emph{competitive ratio} when searching for an oblivious moving target on an $m$-ray star graph. Specifically, we aim to quantify the precise trade-offs between the agents' initial knowledge, communication constraints, and the parity of the graph's rays.

To evaluate how partial information influences performance, we consider three knowledge settings:
\begin{itemize}
    \item \textbf{NoDistance (Known $v$, Unknown $d$):} The robots know the target's constant speed $v < 1$, but have no prior information regarding the target's initial distance $d$ from the origin.
    \item \textbf{NoSpeed (Known $d$, Unknown $v$):} The robots know the target's initial distance $d$ from the origin, but do not know its speed $v$.
    \item \textbf{NoKnowledge (Unknown $d$ and $v$):} The robots possess no prior knowledge about either the target's initial distance $d$ or its speed $v$.
\end{itemize}
\subsection{Context and Motivation}

The main motivation of this study is to understand how 
the lack of knowledge of the target’s speed and/or initial distance—affect the competitive ratio of search algorithms. While search problems have been extensively studied in the literature, most existing work focuses on static targets or simplified settings.

The problem of capturing a moving target has received less attention, especially in non-linear domains and under asymmetric communication models. Previous studies have mainly considered moving target capture on the infinite line, often under simpler communication assumptions such as symmetric or F2F interaction.
Although the $m$-ray star graph has been studied as a search domain, the case of a moving target in this setting remains relatively unexplored. In this work, we aim to address this gap.
We also investigate how different communication models affect performance, showing that the S/R model can improve over the F2F model in certain settings. In addition, we highlight how the parity of the number of rays influences the competitive ratio, illustrating how symmetry breaking in the search domain can lead to improved performance under asymmetric communication.
\paragraph{Motivation for the Sender-Receiver (S/R) Model.}
The Sender-Receiver model models asymmetric communication scenarios commonly encountered in heterogeneous robotic systems. For instance, energy or payload constraints may allow only one agent (the \emph{Sender} $S$) to be equipped with a long-range wireless transmitter or optical beacon, while the other agent (the \emph{Receiver} $R$) carries only a passive receiver or sensor. Studying this model allows us to establish the minimal communication infrastructure required to improve performance over pure face-to-face (F2F) interaction, highlighting how asymmetric signalling interacts with the graph topology.
\subsection{Model, Preliminaries, and Notation}

We consider a star graph consisting of $m$ infinite rays (branches), where $m \geq 2$. The robots are initially located at the origin and can move with unit speed. An oblivious mobile target is placed adversarially on one of the rays at an initial distance $d$ from the origin and moves along that ray with constant speed $v$, either toward or away from the origin.
If the target moves away from the origin, we assume $v < 1$, while if it moves toward the origin, we allow $v \geq 0$. When $v = 0$, the target is stationary. Both robots are equipped with pedometers and computational capabilities, allowing them to infer each other's positions based on exchanged information. The robots and the target start moving simultaneously. Throughout the paper, we use the notation
$
u = \frac{1}{1-v}
\mbox{ and }
M = 
\max\{d,u\}.
$

In the two-robot setting, we consider the following communication models.
Face-to-Face (F2F): Robots can communicate only when they are co-located.
Sender-Receiver (S/R): One robot, called the sender ($S$), can transmit information wirelessly but receives information only via face-to-face communication, while the other robot, called the receiver ($R$), can receive information wirelessly but transmits only via face-to-face communication.
The analysis takes into account the following knowledge models:
NoDistance: the speed $v$ is known, but the initial distance $d$ is unknown;
NoSpeed: the initial distance $d$ is known, but the speed $v$ is unknown;
NoKnowledge: neither $d$ nor $v$ is known.
\paragraph{Formal Definition of Competitive Ratio.}
Let $T_{\mathcal{A}}(d, v)$ denote the time taken by an online search algorithm $\mathcal{A}$ to capture a target placed at distance $d$ moving at speed $v < 1$. An optimal offline algorithm with complete knowledge of $d$ and $v$ moves directly along the target's ray and captures it at time $T_{opt}(d, v) = \frac{d}{1-v} = u \cdot d$, where $u = \frac{1}{1-v}$. The \emph{competitive ratio} $\text{CR}(\mathcal{A})$ is formally defined as:
\begin{equation}
    \text{CR}(\mathcal{A}) = \sup_{d > 0, \, 0 \leq v < 1} \frac{T_{\mathcal{A}}(d, v)}{T_{opt}(d, v)} = \sup_{d > 0, \, 0 \leq v < 1} \frac{T_{\mathcal{A}}(d, v)}{u \cdot d}.
\end{equation}
\subsection{Related Work}

Search and evacuation problems have been studied in the literature under a wide range of models and assumptions. Work by Bellman~\cite{bellman1963optimal} introduced stochastic search models, laying the foundation for subsequent research. Numerous variants have been explored, including settings with static or mobile targets, multiple agents with heterogeneous speeds, and communication capabilities~\cite{beck1973return,bose2017general,bose2013revisiting}. 

In the linear domain, evacuation and search problems with multiple robots have received significant attention. For example, evacuation on a line with two robots of distinct speeds was studied in~\cite{bampas2019linear}, where coordinated zigzag strategies achieve a competitive ratio of $9$ for a wide range of speeds. These results highlight the importance of coordination and synchronization between agents in minimizing search time.

Beyond the line, search and evacuation have also been investigated in more complex environments. Multi-agent search on tree graphs was considered in~\cite{devillez2021two} under both wireless and face-to-face communication models, achieving improved competitive ratios through bidirectional exploration strategies. 
The $m$-ray star graph is a natural extension of the linear domain and has been studied in several contexts. By $m$-ray star we mean a star graph consisting of $m$ rays. For $m=1$, this is an infinite half-line with the robot starting at the unique endpoint and the search problem is trivial. For $m=2$, it is an infinite line with the robot starting at the origin in which case the search problem has been analyzed in the papers such as~\cite{escape2024,coleman2023line}. In conclusion, the interesting unknown case is $m \geq 3$.
Multi-target exploration on $m$-ray graphs was investigated in~\cite{angelopoulos2014multi}, where performance is evaluated under different levels of prior knowledge. Wireless evacuation on $m$-ray graphs with multiple robots was studied in~\cite{brandt2020wireless}, establishing both upper and lower bounds through coordinated exploration strategies.

Search for moving targets has been studied in both stochastic and deterministic settings. McCabe~\cite{mccabe1974} considered stochastic motion, while Alpern and Gal~\cite{gal_search_games} developed a deterministic framework for searching oblivious targets moving at constant speed.
Distributed pursuit and search problems on graph topologies have also been studied in the context of network security and intruder containment; see, e.g., Blin et al.~\cite{blin2008distributed} for distributed chasing of network intruders.
Our work is closely related to these lines of research, but differs in several key aspects. We study the problem of capturing a mobile target in an $m$-ray star graph, considering two-robot settings. In the latter, we analyze two communication models—Face-to-Face (F2F) and Sender-Receiver (S/R)—and investigate how communication constraints, target dynamics, and partial knowledge affect the design and performance of search strategies.

\subsection{Results of the paper}

\begin{table}[H]
\caption{Competitive Ratios (CRs) for an $m$-ray graph with two robots (F2F model)}
\label{StarGraphTwoRobotsF2F}
\centering
\scriptsize
\renewcommand{\arraystretch}{1.2}
\resizebox{\columnwidth}{!}{%
\begin{tabular}{|c|c|c|}
\hline
Knowledge & Parity of $m$ & Competitive Ratio \\
\hline
$v$ & Even &
$1+\dfrac{(m+2)^{\frac{m+2}{2}}(1+v)^{\frac{m}{2}}
u^{\frac{m+2}{2}}}{m^{\frac{m}{2}}}$ \\
\hline
$v$ & Odd &
$1+\dfrac{(m+3)^{\frac{m+3}{2}}(1+v)^{\frac{m+1}{2}}
u^{\frac{m+3}{2}}}{(m+1)^{\frac{m+1}{2}}}$ \\
\hline
$d$ & Even &
$1+O\left(u^{2^{\frac{m+2}{2}}}\log^2(u)\right)$ \\
\hline
$d$ & Odd &
$1+O\left(u^{2^{\frac{m+3}{2}}}\log^2(u)\right)$ \\
\hline
None & Even &
$1+O\left(M^{2^{\frac{m+4}{2}}}
\log^2(M)\log\log(M)\right)$ \\
\hline
None & Odd &
$1+O\left(M^{2^{\frac{m+5}{2}}}
\log^2(M)\log\log(M)\right)$ \\
\hline
\end{tabular}%
}
\end{table}
\begin{table}[t]
\caption{Competitive Ratios (CRs) for an $m$-ray graph with two robots (S/R model)}
\label{StarGraphTwoRobotsSR}
\centering
\scriptsize
\renewcommand{\arraystretch}{1.4}
\begin{tabular}{|c|c|p{0.58\columnwidth}|}
\hline
Knowledge & Parity & Competitive Ratio \\
\hline
$v$ & Even &
$\displaystyle
1+\frac{(m+2)^{\frac{m+2}{2}}(1+v)^{\frac{m}{2}}}
{m^{\frac{m}{2}}}
u^{\frac{m+2}{2}}$ \\
\hline
$v$ & Odd &
$\displaystyle
2u+
\frac{(m-1)(1+v^2)}{1+v}u
\left(
\frac{(m+1)(1+v)u}{m-1}
\right)^{\frac{m+1}{2}}$ \\
\hline
$d$ & Even &
$\displaystyle
1+O\left(
u^{2^{\frac{m+2}{2}}}\log^2(u)
\right)$ \\
\hline
$d$ & Odd &
$\displaystyle
O(u)+O\left(
u^{2^{\frac{m+1}{2}}+1}\log^2(u)
\right)$ \\
\hline
None & Even &
$\displaystyle
1+O\left(
M^{2^{\frac{m+4}{2}}}
\log^2(M)\log\log(M)
\right)$ \\
\hline
None & Odd &
$\displaystyle
O(u)+O\left(
M^{2^{\frac{m+3}{2}}}
\frac{u\log\log M}{d}
\log^2(M)
\right)$ \\
\hline
\end{tabular}
\end{table}
\paragraph{Discussion of Results.}
Tables~\ref{StarGraphTwoRobotsF2F} and~\ref{StarGraphTwoRobotsSR} reveal several structural insights:
\begin{enumerate}
    \item \textbf{Impact of Knowledge:} Knowing the speed $v$ yields constant upper bounds on the competitive ratio (dependent on $m$ and $v$), whereas lacking speed knowledge ($\text{NoSpeed}$ and $\text{NoKnowledge}$) introduces logarithmic overhead factors in terms of $u$ or $M = \max\{d, u\}$.
    \item \textbf{F2F vs. S/R Model:} When $m$ is even, the S/R model provides no asymptotic improvement over F2F due to graph symmetry. However, when $m$ is odd, the asymmetric wireless capability enables symmetry breaking, reducing the exponent in $u$ from $\frac{m+3}{2}$ to $\frac{m+1}{2}$.
\end{enumerate}
A conference version of this paper is appearing in the proceedings of ICDCN~\cite{jawhar2026ICDCN}.
\section{Searching for a Moving Target Using the F2F model}
\subsection{Overview of the Search Strategy}
All proposed search algorithms share a common geometric structure: the two robots operate in synchronized rounds. In round $k$, the robots explore two consecutive, previously unexamined rays up to a specific search depth $x_k$ (or $a_k$), before returning to the origin to exchange information or synchronize. If neither robot encounters the target, they advance to the next pair of rays and scale up their search radius according to an exponentially increasing parameter sequence.
\subsection{NoDistance Away Model}
The idea is to have both robots move the same distance in each round. In every round, the two robots explore two consecutive rays. Each robot moves a distance $a^r$ during round $r$, where $a$ is a positive constant, and then returns to the origin. If one of the robots finds the target, it moves along with the target. If neither robot finds the target, both robots return to the origin, proceed to the next pair of consecutive rays, and increment the distance to $a^{r+1}$ for the next round. 

Assume we have two robots, $R_{1}$ and $R_{2}$, initially located at the origin, and a target placed at a distance $d$ from the origin. The algorithm proceeds as follows:

\begin{algorithm}[H]
\caption{NDAwayMrayWithTwoRobotsF2F}
\label{NDAwayMrayWithTwoRobotsF2F}
\begin{algorithmic}[1]
\For{$k \gets 1$ to $\infty$}
    \State $R_{1}$ moves along one ray to distance $a^k$ (unless the target is found), then returns to the origin.
    \State $R_{2}$ moves along the consecutive ray to distance $a^k$ (unless the target is found), then returns to the origin.
    \If{Target is found by $R_{1}$}
        \State $R_{1}$ switches direction and moves to catch $R_{2}$.
        \State Both robots move together to catch up with the target.
        \State \textbf{terminate}
    \ElsIf{Target is found by $R_{2}$}
        \State $R_{2}$ switches direction and moves to catch $R_{1}$.
        \State Both robots move together to catch up with the target.
        \State \textbf{terminate}
    \Else
        \State $R_{1}$ and $R_{2}$ proceed to the next two consecutive rays.
    \EndIf
\EndFor
\end{algorithmic}
\end{algorithm}

We prove the following result.

\begin{theorem}
\label{NDAwayMrayWithTwoRobotsF2FTheorem}
The competitive ratio of Algorithm~\ref{NDAwayMrayWithTwoRobotsF2F} is at most
\[
1+\frac{(m+2)^{\frac{m+2}{2}}(1+v)^{\frac{m}{2}}\cdot u^{\frac{m+2}{2}}}{m^{\frac{m}{2}}}, \quad \text{if $m$ is even},
\]
and
\[
1+\frac{(m+3)^{\frac{m+3}{2}} (1+v)^{\frac{m+1}{2}}\cdot u^{{\frac{m+3}{2}}}}{(m+1)^{\frac{m+1}{2}}}, \quad \text{if $m$ is odd}.
\]
\end{theorem}
\begin{proof} (Theorem~\ref{NDAwayMrayWithTwoRobotsF2FTheorem})
The proof consists of two cases.
\noindent
\textbf{Case 1: $m$ is even.}
Suppose the target is missed in round $r-1$ on a ray. To miss the target in round $r-1$, the target must have reached a distance greater than the robot's reach $a_{r-1}$. Since the target moves at speed $v$, its distance at this time is $d +2v \sum_{i=0}^{r-2} a_i$. Then, for $R_1$ to miss the target in round $r-1$, we have:
\[
u(d + 2v\sum_{i=0}^{r-2} a^i) \geq a^{r-1} \implies
d \geq \frac{a^r(1-v) - a^{r-1}(1+v) + 2v}{a-1}.
\]
If the target is missed by one of the robots in round \(r-1\), it will be captured in round \((r-1)+\frac{m}{2}\). Suppose that robot \(R_1\) captures the target. The time required for \(R_1\) to catch up with \(R_2\) is \(a^{(r-1)+\frac{m}{2}}\). Therefore, the total time for both robots to capture the target is given by:
$
u\cdot (x_{(r-1)+\frac{m}{2}} + vx_{(r-1)+\frac{m}{2}}).
$
Thus, the competitive ratio satisfies:
\begin{align*}
CR &\leq 1 + \frac{2\sum_{i=0}^{(r-2)+\frac{m}{2}} a^i + 2x_{(r-1)+\frac{m}{2}}}{d}= 1 + \frac{2a^{r+\frac{m}{2}} - 2}{(a - 1)d}.
\end{align*}

Substituting the bound on $d$ gives:
\[
CR \leq 1 + \frac{2a^{r+\frac{m}{2}}}{a^{r-1}(a(1-v) - (1+v))} = 1 + \frac{2a^{\frac{m+2}{2}}}{(1-v)a - (1+v)}.
\]
This expression is minimized when:
$
a = u\cdot \frac{(m+2)(1+v)}{m}.
$
Substituting this into the competitive ratio gives:
\begin{align*}
CR &\leq 1 + \frac{\frac{2(m+2)^{\frac{m+2}{2}}(1+v)^{\frac{m+2}{2}}}{m^{\frac{m+2}{2}}(1-v)^{\frac{m+2}{2}}}}{\frac{(m+2)(1+v)}{m} - (1+v)}=1+\frac{(m+2)^{\frac{m+2}{2}}(1+v)^{\frac{m}{2}}\cdot u^{\frac{m+2}{2}}}{m^{\frac{m}{2}}}.
\end{align*}

To confirm the result, let $m = 2$ and substitute into the formula. We obtain:
$
CR = 1 + u^2\cdot 8(1+v),
$
which agrees with the result for the single-robot case reported in \cite{coleman2023line}.

\noindent
\textbf{Case~2: $m$ is odd.} 
Suppose the target is missed at round $r-1$ by robot $R_1$. Then, for $R_1$ to capture the target, the required time is:
\[
2\sum_{i=0}^{(r-2)+\frac{m+1}{2}} a^i + u\left(d + 2v\sum_{i=0}^{(r-2)+\frac{m+1}{2}} a^i\right).
\]

Since $R_1$ misses the target in round $r-1$, we have:
\[
u\left(d + 2v\sum_{i=0}^{r-2} a^i\right) \geq a^{r-1} \implies
d \geq \frac{a^r(1-v) - a^{r-1}(1+v) + 2v}{a-1}.
\]

The time needed for $R_1$ to catch $R_2$ is $a^{(r-1)+\frac{m+1}{2}}$, and the time for both robots to catch the target is:
$
u\left(x_{(r-1)+\frac{m+1}{2}} + vx_{(r-1)+\frac{m+1}{2}}\right).
$
Thus, the competitive ratio satisfies:
\begin{align*}
CR &\leq 1 + \frac{2\sum_{i=0}^{(r-2)+\frac{m+1}{2}} a^i + 2x_{(r-1)+\frac{m+1}{2}}}{d}= 1 + \frac{2a^{r+\frac{m+1}{2}}- 2}{(a - 1)d}.
\end{align*}

Substituting the bound on $d$ gives:
\[
CR \leq 1 + \frac{2a^{r+\frac{m+1}{2}}}{a^{r-1}(a(1-v) - (1+v))} = 1 + \frac{2a^{\frac{m+3}{2}}}{(1-v)a - (1+v)}.
\]

This expression is minimized when:
$
a = \frac{(m+3)(1+v)}{(m+1)(1-v)}.
$
Substituting this into the competitive ratio gives:
\begin{align*}
CR \leq 1+\frac{(m+3)^{\frac{m+3}{2}} (1+v)^{\frac{m+1}{2}}\cdot u^{{\frac{m+3}{2}}}}{(m+1)^{\frac{m+1}{2}}}.
\end{align*}
This proves Theorem~\ref{NDAwayMrayWithTwoRobotsF2FTheorem}.
\qed
\end{proof}

\subsection{NoSpeed Away Model}
The idea is to have both robots move in a v-shape to cover the same distance in each round. Let us consider a monotone increasing sequence of non-negative integers $\{f_{i}:i\geq 0\}$. Each robot tries to guess the speed of the target. We will use the guess $v_{i}=1-2^{-f_{i}}$. In every round, the two robots explore two consecutive rays and try to find the target by guessing its speed. Each robot moves a distance determined by the speed guess $v_i$, ensuring that if the target’s speed is less than or equal to $v_i$, the target will be captured during that phase. If the target is not caught during iteration $i$, then both robots move to two new consecutive rays and increase the guess to $v_{i+1}$.The robots repeat this process until one of them captures the target. The other robot will detect that the target has been captured when it does not find its partner at the origin, and it will then proceed toward the target. The algorithm is as follows:

\begin{algorithm}[H]
\caption{NSAwayMrayWithTwoRobotsF2F}
\label{NSAwayMrayWithTwoRobotsF2F}
Input: Target Initial distance d
 and Integer increasing sequence $\{f_{i}:i\geq 0\}$ with $f_{0}=1$, $t=0$ and $i\geq 0$
\begin{algorithmic}[1]
\For{$k \gets 1$ to $\infty$}
    \State $v_{i}=1-2^{-f_{i}}$.
    \State $R_{1}$ moves along one ray to distance $x_{i}=(-1)^{i} \frac{d+tv_{i}}{1-v_{i}}$ (unless the target is found), then returns to the origin.
    \State $R_{2}$ moves along the consecutive ray to distance $x_{i}=(-1)^{i} \frac{d+tv_{i}}{1-v_{i}}$ (unless the target is found), then returns to the origin.
    \If{Target is found by $R_{1}$}
        \State $R_{1}$ proceeds with the target.
        \State \textbf{terminate}
    \ElsIf{Target is found by $R_{2}$}
        \State $R_{2}$ proceeds with the target.
        \State \textbf{terminate}
    \Else
        \State $R_{1}$ and $R_{2}$ proceed to the next two consecutive rays.
    \EndIf
    \If{One robot does not meet the other at the origin}
         \State It proceeds along the ray currently being explored by the other robot until it reaches the target.
        \State \textbf{terminate}
    \EndIf
    \State $t=t+2x_{i}$
\EndFor
\end{algorithmic}
\end{algorithm}

We prove the following result.
\begin{theorem}
\label{NSAwayMrayWithTwoRobotsF2FTheorem}
The competitive ratio of Algorithm~\ref{NSAwayMrayWithTwoRobotsF2F} is at most
\[
1 + O\left( u^{2^{\frac{m+3}{2}}} \log^2(u) \right), \quad \text{if $m$ is odd},
\]
and
\[
1 + O\left( u^{2^{\frac{m+2}{2}}} \log^2(u) \right), \quad \text{if $m$ is even}.
\]
\end{theorem}
\begin{proof} (Theorem~\ref{NSAwayMrayWithTwoRobotsF2FTheorem})
We have 2 cases to consider. Either the number of rays $m$ is even or odd. 

\noindent    
     Let $d_i$ be the distance from the origin to where the target would be if its speed was equal to $v_i$. Recall the algorithm involves the robot moving a distance $x_i$ (in time $x_i$, since the robot's speed is $1$) away from the origin and back in round $i$.  Observe then that $d_0 = d$, $v_0 = 1/2$, and
\begin{align}
    d_i = d+ 2v_i \sum_{j=0}^{i-1}  x_j.
\end{align}
Therefore, it follows from the definition of $x_i$ that
  \begin{align}
    x_i = 2^{f_i} \left( d+ 2v_i \sum_{j=0}^{i-1}  x_j \right) . 
  \end{align}
  As a consequence
  \begin{align}
    \sum_{j=0}^{i-1}  x_j = \frac{x_i - 2^{f_i} d}{2^{f_i} \cdot 2 \cdot v_i} \label{eq:no_speed_upper_guess_3}
  \end{align}
  Similarly, if we replace $i$ with $i+1$ we have that
  \begin{align}
    \sum_{j=0}^{i} x_j = \frac{x_{i+1} - 2^{f_{i+1}} d}{2^{f_{i+1}} \cdot 2 \cdot v_{i+1}} \label{eq:no_speed_upper_guess_4}
  \end{align}
  Subtracting Equation~\eqref{eq:no_speed_upper_guess_3} from  Equation~\eqref{eq:no_speed_upper_guess_4}, we derive the recurrence
  \begin{align}
    x_i = \frac{x_{i+1} - 2^{f_{i+1}} d}{2^{f_{i+1}} \cdot 2 \cdot v_{i+1}} - \frac{x_i - 2^{f_i} d}{2^{f_i} \cdot 2 \cdot v_i} \label{eq:no_speed_upper_guess_5} 
  \end{align}
  Collecting similar terms and simplifying Equation~\eqref{eq:no_speed_upper_guess_5} yields
  \begin{align}
    \frac{x_{i+1}}{2^{f_{i+1}} \cdot 2 \cdot v_{i+1}}
      &= \notag \left( 1 + \frac 1{2^{f_i} \cdot 2 \cdot v_i}\right) x_i
      + \left( \frac{2^{f_{i+1}}}{2^{f_{i+1}} \cdot 2 \cdot v_{i+1}} - 
      \frac{2^{f_i}}{2^{f_i} \cdot 2 \cdot v_{i}} \right) d \\
    &= \label{eq:no_speed_upper_guess_5_inter} x_i \left( 1+\frac{1}{2^f_i \cdot 2 \cdot v_i} \right) + \frac{d}{2 v_{i+1}} - \frac{d}{2 v_{i}}  \\
    &\leq \label{eq:no_speed_upper_guess_5_new} \left( 1 + \frac 1{2^{f_i} \cdot 2 \cdot v_i}\right) x_i  
  \end{align}
  following from the fact the sum of the last two terms in Inequality~\eqref{eq:no_speed_upper_guess_5_inter} is less than or equal to $0$.

 If we simplify the right-hand side of Equation~\eqref{eq:no_speed_upper_guess_5_new}, we derive the following recursive inequalities
  \begin{align}
    x_{i+1} &\leq 2^{f_{i+1}} \cdot 2 \cdot v_{i+1} \left( 1+ \frac 1{2^{f_i} \cdot 2 \cdot v_i}\right) x_i \notag \\
    &\leq 2^{f_{i+1}} \cdot 2 \cdot \left( 1+ \frac 1{2^{f_i}}\right) x_i \notag \leq \left( 2 \cdot 2^{f_{i+1}} + 2\cdot 2^{f_{i+1} - f_i}\right) x_i \notag
    \\
    &\leq 2^{f_{i+1} } \cdot 4 \cdot x_i , \label{eq:no_speed_upper_guess_6} 
  \end{align}
  which follows since $\frac{1}{2} \leq v_i < 1$ and $f_i < f_{i+1}$ for all $i$.

 By repeated application of the last Recurrence~\eqref{eq:no_speed_upper_guess_6} above and using the fact that by definition $x_0 = 2^{f_0}d$, it follows easily by induction that
  \begin{align}
    x_{i+1} 
    &\leq \notag 2^{f_{i+1} } \cdot 4 \cdot x_i \\
    &\leq \notag 2^{f_{i+1}+f_i} \cdot 4^2 \cdot x_{i-1}  \\
    &~ \vdots \notag \\
    &\leq \notag 2^{f_{i+1}+f_i+f_{i-1}+ \cdots + f_1} \cdot 4^{i+1} \cdot x_0  \\
    &\leq \label{eq:no_speed_upper_guess_7} 2^{\sum_{j=0}^{i+1} f_j} \cdot 4^{i+1} \cdot d 
  \end{align}

  using \ref{eq:no_speed_upper_guess_3} and \ref{eq:no_speed_upper_guess_7}.
  Assume that the target is missed at some round $r-1$, the competitive ratio becomes as follows:
  \begin{align*}
      CR&=\frac{2\sum_{i=0}^{(r-1)+\frac{m+1}{2}} x_{i} +u(d+2v \sum_{i=0}^{(r-1)+\frac{m+1}{2}} x_{i})}{du}\\
      &\leq 1+\frac{2 \sum_{i=0}^{(r-1)+\frac{m+1}{2}} x_{i}}{d}\leq 1+\frac{2}{d}\left(\frac{x_{r+\frac{m+1}{2}}-2^{f_{r+\frac{m+1}{2}}}d}{2\cdot 2^{f_{r+\frac{m+1}{2}}}v_{r+\frac{m+1}{2}}}\right)\\
      &\leq 1+\frac{2}{d}\left(\frac{x_{r+\frac{m+1}{2}}}{2\cdot 2^{f_{r+\frac{m+1}{2}}}v_{r+\frac{m+1}{2}}}\right)\leq 1+\frac{2^{ \sum_{i=0}^{r+\frac{m+1}{2}} f_{i}}\cdot 4^{r+\frac{m+1}{2}}}{v_{r+\frac{m+1}{2}}2^{f_{r+\frac{m+1}{2}}}}\\
      &\leq 1+2^{ 1+\sum_{i=0}^{r+\frac{m-1}{2}} f_{i}}\cdot 4^{r+\frac{m+1}{2}}\\
  \end{align*}
We consider $f_r = 2^r$. Since the target is missed at iteration $r-1$, we have
\[
1 - 2^{-f_{r-1}} < v, \implies
2^{f_{r-1}} < u.
\]
Substituting $f_{r-1} = 2^{r-1}$, we obtain
$
2^{2^{r-1}} < u.
$
Thus, the competitive ratio becomes as follows:
  \begin{align*}
      CR&\leq 1+2^{1+ \sum_{i=0}^{r+\frac{m-1}{2}} 2^{i}}\cdot 4^{r+\frac{m+1}{2}}\leq 1+2^{2^{r+\frac{m+1}{2}}}\cdot 4^{r+\frac{m+1}{2}}\\
      &\leq 1+u^{2^{\frac{m+3}{2}}}\cdot 4^{r+\frac{m+1}{2}}\leq 1+u^{2^{\frac{m+3}{2}}}\cdot 4^{\frac{m+3}{2}}\log^2(u)
  \end{align*}

\noindent  
{\bf Case~2: $m$ is even.}
  Using equation~(\ref{eq:no_speed_upper_guess_3}):
  \begin{align}
    \sum_{j=0}^{i-1}  x_j = \frac{x_i - 2^{f_i} d}{2^{f_i} \cdot 2 \cdot v_i} 
  \end{align}
  and equation~(\ref{eq:no_speed_upper_guess_7}):
  \begin{align}
   x_{i+1}\leq 2^{\sum_{j=0}^{i+1} f_j} \cdot 4^{i+1} \cdot d 
  \end{align}
  Assume that the target is missed at some round $r-1$, the competitive ratio becomes as follows:
  \begin{align*}
      CR&=\frac{\sum_{i=0}^{(r-1)+\frac{m}{2}} 2x_{i} +u\left(d+2v \sum_{i=0}^{(r-1)+\frac{m}{2}} x_{i}\right)}{ud}\leq 1+\frac{2 \sum_{i=0}^{(r-1)+\frac{m}{2}} x_{i}}{d}\\
      &\leq 1+\frac{2}{d}\left(\frac{x_{r+\frac{m}{2}}-2^{f_{r+\frac{m}{2}}}d}{2\cdot 2^{f_{r+\frac{m}{2}}}v_{r+\frac{m}{2}}}\right)\leq 1+\frac{2}{d}\left(\frac{x_{r+\frac{m}{2}}}{2\cdot 2^{f_{r+\frac{m}{2}}}v_{r+\frac{m}{2}}}\right)\\
      &\leq 1+\frac{2^{ \sum_{i=0}^{(r-1)+\frac{m}{2}} f_{i}}\cdot 4^{r+\frac{m}{2}}}{v_{r+\frac{m}{2}}}\leq 1+2^{ 1+\sum_{i=0}^{r+\frac{m-2}{2}} f_{i}}\cdot 4^{r+\frac{m}{2}}\\
  \end{align*}
  We consider $f_r = 2^r$. Since the target is missed at iteration $r-1$, we have
\[
1 - 2^{-f_{r-1}} < v, \implies
2^{f_{r-1}} < u.
\]
Substituting $f_{r-1} = 2^{r-1}$, we obtain
$
2^{2^{r-1}} < u.
$
Thus, the competitive ratio becomes as follows:
  \begin{align*}
      CR&\leq 1+2^{ 1+\sum_{i=0}^{r+\frac{m-2}{2}} 2^{i}}\cdot 4^{r+\frac{m}{2}}\leq 1+2^{2^{r+\frac{m}{2}}}\cdot 4^{r+\frac{m}{2}}\\
      &\leq 1+u^{2^{\frac{m+2}{2}}}\cdot 4^{r+\frac{m}{2}}\\
      &\leq 1+u^{2^{\frac{m+2}{2}}}\cdot 4^{\frac{m+2}{2}}\log^2(u)
  \end{align*}
This proves Theorem~\ref{NSAwayMrayWithTwoRobotsF2FTheorem}.
\qed  
\end{proof}
\subsection{NoKnowledge Away Model}
The idea is to have both robots move in a v-shape to cover the same distance in each round. Let us consider two monotone increasing sequences of non-negative integers $\{f_{i}:i\geq 0\}$ and $\{g_{i}:i\geq 0\}$. Each robot tries to guess the speed and the initial distance of the target. We will use the guess $v_{i}=1-2^{-f_{i}}$ for the speed and $d_{i}=2^{g_{i}}$ for the distance. In every round, the two robots explore two consecutive rays and try to find the target by guessing its speed and distance. Each robot moves a distance determined by the speed guess $v_i$ and the distance guess $d_{i}$, ensuring that if the target’s speed is less than or equal to $v_i$ and the target's distance is less than or equal to $d_{i}$, the target will be captured during that phase. If the target is not caught during iteration $i$, then both robots move to two new consecutive rays and increase the guess to $v_{i+1}$ and $d_{i+1}$. The robots repeat this process until one of them captures the target. The other robot will detect that the target has been captured when it does not find its partner at the origin, and it will then proceed toward the target. The algorithm is as follows:

\begin{algorithm}[H]
\caption{NKAwayMrayWithTwoRobotsF2F}
\label{NKAwayMrayWithTwoRobotsF2F}
Input: Integer increasing sequences $\{f_{i}:i\geq 0\}$ and $\{g_{i}:i\geq 0\}$ with $f_{0}=1$, $g_{0}=0$, $t=0$ and $i\geq 0$
\begin{algorithmic}[1]
\For{$k \gets 1$ to $\infty$}
    \State set $v_{i}=1-2^{-f_{i}}$ and $d_{i}=2^{g_{i}}$
    \State $R_{1}$ moves along one ray to distance $x_{i}=(-1)^{i} \frac{d_{i}+tv_{i}}{1-v_{i}}$ (unless the target is found), then returns to the origin.
    \State $R_{2}$ moves along the consecutive ray to distance $x_{i}=(-1)^{i} \frac{d_{i}+tv_{i}}{1-v_{i}}$ (unless the target is found), then returns to the origin.
    \If{Target is found by $R_{1}$}
        \State $R_{1}$ proceeds with the target.
        \State \textbf{terminate}
    \ElsIf{Target is found by $R_{2}$}
        \State $R_{2}$ proceeds with the target.
        \State \textbf{terminate}
    \Else
        \State $R_{1}$ and $R_{2}$ proceed to the next two consecutive rays.
    \EndIf
    \If{One robot does not meet the other at the origin}
         \State It proceeds along the ray currently being explored by the other robot until it reaches the target.
        \State \textbf{terminate}
    \EndIf
    \State $t=t+2x_{i}$
\EndFor
\end{algorithmic}
\end{algorithm}

We prove the following result.

\begin{theorem}
\label{NKAwayMrayWithTwoRobotsF2FTheorem}
The competitive ratio of Algorithm~\ref{NKAwayMrayWithTwoRobotsF2F} is at most
\[
1 + O\left( M^{2^{\frac{m+4}{2}}} \log^2(M) \log \log(M) \right), \quad \text{if $m$ is even},
\]
and
\[
1 + O\left( M^{2^{\frac{m+5}{2}}} \log^2(M) \log \log(M) \right), \quad \text{if $m$ is odd}.
\]
\end{theorem}
\begin{proof} (Theorem~\ref{NKAwayMrayWithTwoRobotsF2FTheorem})
We have 2 cases to consider. The number of rays $m$ is even or odd. 

\noindent
{\bf Case~1: $m$ is odd.}
     Let $d_i$ be the distance from the origin to where the target would be if its speed was equal to $v_i$ and its starting position $2^{g_i}$. Observe from the algorithm that $d_0 = 1$, $v_0 = 1/2$, and
\begin{align}
    d_i = 2^{g_{i}}+ 2v_i \sum_{j=0}^{i-1}  x_j. \label{eq:no_knowledge_upper_guess_d}
\end{align}
Therefore, it follows from the definition of $x_i$ that
  \begin{align}
    x_i = 2^{f_i} \left( 2^{g_{i}}+ 2v_i \sum_{j=0}^{i-1}  x_j \right) . \label{eq:no_knowledge_upper_guess_x}
  \end{align}
  As a consequence
  \begin{align}
    \sum_{j=0}^{i-1}  x_j = \frac{x_i - 2^{f_i} 2^{g_{i}}}{2^{f_i} \cdot 2 \cdot v_i} \label{eq:no_knowledge_upper_guess_1}
  \end{align}
  Similarly, if we replace $i$ with $i+1$ we have that
  \begin{align}
    \sum_{j=0}^{i} x_j = \frac{x_{i+1} - 2^{f_{i+1}} 2^{g_{i+1}}}{2^{f_{i+1}} \cdot 2 \cdot v_{i+1}} \label{eq:no_knowledge_upper_guess_2}
  \end{align}
  Subtracting Equation~\eqref{eq:no_knowledge_upper_guess_1} from  Equation~\eqref{eq:no_knowledge_upper_guess_2}, we derive the recurrence
  \begin{align}
    x_i = \frac{x_{i+1} - 2^{f_{i+1}} 2^{g_{i+1}}}{2^{f_{i+1}} \cdot 2 \cdot v_{i+1}} - \frac{x_i - 2^{f_i} 2^{g_{i}}}{2^{f_i} \cdot 2 \cdot v_i} \label{eq:no_knowledge_upper_guess_5} 
  \end{align}

   Collecting similar terms and simplifying Equation~\eqref{eq:no_knowledge_upper_guess_5} yields
  \begin{align}
    \frac{x_{i+1}}{2^{f_{i+1}} \cdot 2 \cdot v_{i+1}}
      &= \notag \left( 1 + \frac 1{2^{f_i} \cdot 2 \cdot v_i}\right) x_i
      + \left( \frac{2^{f_{i+1}+g_{i+1}}}{2^{f_{i+1}} \cdot 2 \cdot v_{i+1}} - 
      \frac{2^{f_i+g_{i}}}{2^{f_i} \cdot 2 \cdot v_{i}} \right) \\
    &= \label{eq:no_knowledge_upper_guess_5_inter} x_i \left( 1+\frac{1}{2^f_i \cdot 2 \cdot v_i} \right) + 2^{g_{i+1}-1}(\frac{1}{v_{i+1}} - \frac{2^{g_{i}}}{2^{g_{i+1}} v_{i}})  \\
    &\leq \label{eq:no_knowledge_upper_guess_5_new} \left( 1 + \frac 1{2^{f_i} \cdot 2 \cdot v_i}\right) x_i+2^{g_{i+1}-1}  
  \end{align}
  Following from the fact the sum of the last two terms in Inequality~\eqref{eq:no_knowledge_upper_guess_5_inter} is less than or equal to $1$.

  If we simplify the right-hand side of Equation~\eqref{eq:no_knowledge_upper_guess_5_new}, we derive the following recursive inequalities
  \begin{align}
    x_{i+1} &\leq 2^{f_{i+1}} \cdot 2 \cdot v_{i+1} \left( 1+ \frac 1{2^{f_i} \cdot 2 \cdot v_i}\right) x_i +2^{f_{i+1}+g_{i+1}}v_{i+1}\notag \\
    &\leq \left( 2 \cdot (2^{f_{i+1}}-1) + 2^{f_{i+1} - f_i}\cdot \frac{v_{i+1}}{v_{i}}\right) x_i+2^{f_{i+1}+g_{i+1}} \notag
    \\
    &\leq 2^{f_{i+1} } \cdot 4 \cdot x_i+2^{f_{i+1}+g_{i+1}} \label{eq:no_knowledge_upper_guess_6} 
  \end{align}
  which follows since $\frac{v_{i+1}}{v_{i}} \leq 2$.
 By repeated application of the last Recurrence~\eqref{eq:no_knowledge_upper_guess_6} above and using the fact that by definition $x_0 = 2^{f_0}d$, it follows easily by induction that
  \begin{align}
    x_{i+1} 
    &\leq \notag 2^{f_{i+1} } \cdot 4 \cdot x_i+2^{f_{i+1}+g_{i+1}} \\
    &\leq \notag 2^{f_{i+1}+f_{i}} \cdot 4^2 \cdot x_{i-1}+ 2^{f_{i+1}+f_{i}+g_{i}}\cdot 4+ 2^{f_{i+1}+g_{i+1}}\\
    &~ \vdots \notag \\
    &\leq \label{eq:no_knowledge_upper_guess_7} \sum_{k=0}^{i+1} 2^{g_{k}+\sum_{j=k}^{i+1}f_{j}}\cdot 4^{i-k+1}
  \end{align}
  Assume that the target is missed at some round $r-1$, using \ref{eq:no_knowledge_upper_guess_1} and \ref{eq:no_knowledge_upper_guess_7},
  the competitive ratio becomes as follows:
  \begin{align*}
      CR&\leq \frac{\sum_{i=0}^{(r-1)+\frac{m+1}{2}} 2x_{i} +u\left(d+2v \sum_{i=0}^{(r-1)+\frac{m+1}{2}} x_{i}\right)}{ud}\\
      &\leq 1+\frac{2 \sum_{i=0}^{(r-1)+\frac{m+1}{2}} x_{i}}{d}\\
      &\leq 1+\frac{1}{d}\left(\frac{x_{r+\frac{m+1}{2}}}{2\cdot 2^{f_{r+\frac{m+1}{2}}}v_{r+\frac{m+1}{2}}}\right)\\
      &\leq 1+\frac{\sum_{k=0}^{r+\frac{m+1}{2}} 2^{g_{k}+\sum_{j=k}^{r+\frac{m+1}{2}}f_{j}}\cdot 4^{r+\frac{m+1}{2}-k}}{d\cdot 2^{f_{r+\frac{m+1}{2}}}}\\
      &\leq  1+\frac{\sum_{k=0}^{r+\frac{m+1}{2}} 2^{g_{k}+\sum_{j=k}^{r+\frac{m-1}{2}}f_{j}}\cdot 4^{r+\frac{m+1}{2}-k}}{d}\\
       &\leq 1+\frac{(r+\frac{m+3}{2}) 2^{g_{r+\frac{m+1}{2}}}\cdot 2^{\sum_{j=0}^{r+\frac{m-1}{2}}f_{j}}\cdot 4^{r+\frac{m+1}{2}}}{d}\\
  \end{align*}
 We consider $f_r = g_r = 2^r$. Since the target is missed at iteration $r-1$, we have
$
1 - 2^{-f_{r-1}} < v,
$
which implies
$
2^{f_{r-1}} < u.
$
Substituting $f_{r-1} = 2^{r-1}$, we obtain
$
2^{2^{r-1}} < u.
$
The same thing applies for the distance $d$, we obtain
$
2^{2^{r-1}} < d.
$
Thus, the competitive ratio becomes as follows:
  \begin{align*}
      CR&\leq 1+\frac{(r+\frac{m+3}{2}) 2^{g_{r+\frac{m+1}{2}}}\cdot 2^{\sum_{j=0}^{r+\frac{m-1}{2}}f_{j}}\cdot 4^{r+\frac{m+1}{2}}}{d}\\
      &\leq 1+\frac{(r+\frac{m+3}{2}) 2^{2^{r+\frac{m+1}{2}}}\cdot 2^{2^{r+\frac{m+1}{2}}-1}\cdot 4^{r+\frac{m+1}{2}}}{d}\\
      &\leq 1+\frac{4\cdot((r-1)+\frac{m+5}{2})\cdot (2^{2^{r-1}})^{2^{\frac{m+5}{2}}}\cdot 4^{\frac{m+1}{2}}\cdot4^{r-1}}{d}\\
       &\leq 1+\frac{4^{\frac{m+3}{2}}}{d}\cdot (\log\log M+\frac{m+5}{2})\cdot M^{2^{\frac{m+5}{2}}}\cdot \log^{2}(M)
  \end{align*}

\noindent  
{\bf Case~2: $m$ is even.}
  Using equation~(\ref{eq:no_knowledge_upper_guess_1}):
   \begin{align}
     \sum_{j=0}^{i-1}  x_j = \frac{x_i - 2^{f_i} 2^{g_{i}}}{2^{f_i} \cdot 2 \cdot v_i}
  \end{align}
  and equation~(\ref{eq:no_knowledge_upper_guess_7}):
  \begin{align}
   x_{i+1}\leq \sum_{k=0}^{i+1} 2^{g_{k}+\sum_{j=k}^{i+1}f_{j}}\cdot 4^{i-k+1}
  \end{align}
  Assume that the target is missed at some round $r-1$, the competitive ratio becomes as follows:
  \begin{align*}
      CR&=\frac{\sum_{i=0}^{(r-1)+\frac{m}{2}} 2x_{i} +u\left(d+2v \sum_{i=0}^{(r-1)+\frac{m}{2}} x_{i}\right)}{ud}\leq 1+\frac{2 \sum_{i=0}^{(r-1)+\frac{m}{2}} x_{i}}{d}\\
      &\leq 1+\frac{1}{d}\left(\frac{x_{r+\frac{m}{2}}}{2\cdot 2^{f_{r+\frac{m}{2}}}v_{r+\frac{m+1}{2}}}\right)
      \leq 1+\frac{\sum_{k=0}^{r+\frac{m}{2}} 2^{g_{k}+\sum_{j=k}^{r+\frac{m}{2}}f_{j}}\cdot 4^{r+\frac{m}{2}-k}}{2\cdot d\cdot v_{r+\frac{m}{2}}2^{f_{r+\frac{m}{2}}}}\\
      &\leq  1+\frac{\sum_{k=0}^{r+\frac{m}{2}} 2^{g_{k}+\sum_{j=k}^{r+\frac{m-2}{2}}f_{j}}\cdot 4^{r+\frac{m}{2}-k}}{2\cdot d\cdot v_{r+\frac{m}{2}}}\\
     &\leq 1+\frac{(r+\frac{m+2}{2}) 2^{g_{r+\frac{m}{2}}}\cdot 2^{\sum_{j=0}^{r+\frac{m-2}{2}}f_{j}}\cdot 4^{r+\frac{m}{2}}}{d}\\
  \end{align*}
  We consider $f_r = g_r = 2^r$. Since the target is missed at iteration $r-1$, we have
$
1 - 2^{-f_{r-1}} < v,
$
which implies
$
2^{f_{r-1}} < u.
$
Substituting $f_{r-1} = 2^{r-1}$, we obtain
$
2^{2^{r-1}} < u.
$

The same thing applies for the distance $d$, we obtain
$
2^{2^{r-1}} < d.
$
Thus, the competitive ratio becomes as follows:
  \begin{align*}
      CR&\leq 1+\frac{(r+\frac{m+2}{2}) 2^{g_{r+\frac{m}{2}}}\cdot 2^{\sum_{j=0}^{r+\frac{m-2}{2}}f_{j}}\cdot 4^{r+\frac{m}{2}}}{d}\\
      &\leq 1+\frac{(r+\frac{m+2}{2}) 2^{2^{r+\frac{m}{2}}}\cdot 2^{2^{r+\frac{m}{2}-1}}\cdot 4^{r+\frac{m}{2}}}{d}\\
      &\leq 1+\frac{4\cdot((r-1)+\frac{m+4}{2})\cdot (2^{2^{r-1}})^{2^{\frac{m+4}{2}}}\cdot 4^{\frac{m}{2}}\cdot 4^{r-1}}{d}\\
       &\leq 1+\frac{4^{\frac{m+2}{2}}}{d}\cdot (\log\log M+\frac{m+4}{2})\cdot M^{2^{\frac{m+4}{2}}}\cdot \log^{2}(M)
  \end{align*}
This proves Theorem~\ref{NKAwayMrayWithTwoRobotsF2FTheorem}.
~\qed
\end{proof}
\section{Searching for a Moving Target Using the S/R model}
In this section, we propose an algorithm for the S/R communication model that achieves a better upper bound than the F2F model when the number of rays is odd. This follows from the fact that the symmetry of the search configuration cannot be broken in the even-ray case. For an odd number of rays, the search process requires $(m+1)/2$ iterations before the target can be captured. This additional iteration gives the sender another opportunity to encounter the target and communicate its location to the receiver before the exploration pattern is completed, resulting in an improved competitive ratio. In contrast, when the number of rays is even, every ray has a corresponding opposite ray, and the initial configuration remains symmetric. Since deterministic algorithms cannot break this symmetry, the sender cannot obtain an advantage by communicating the target location earlier than in the F2F model. Therefore, the S/R communication model does not provide an improvement over F2F when the number of rays is even. We formalize this observation in the following lemma.
\begin{lemma}[Symmetry Invariance for Even $m$]
\label{lem:even_m_symmetry}
For any even integer $m$, any deterministic algorithm $\mathcal{A}_{SR}$ under the Sender-Receiver model can be simulated by a deterministic algorithm $\mathcal{A}_{F2F}$ under the Face-to-Face model with identical competitive ratio.
\end{lemma}
\begin{proof}
When $m$ is even, the star graph exhibits central symmetry around the origin (every ray has a diametrically opposed counterpart). Since both robots start at the origin and follow deterministic strategies, the adversary can exploit this symmetry to place the target on a ray that the receiver visits before the sender can communicate the target's location. Moreover, because the number of rays is even, the robot roles repeat after a complete cycle of ray exploration. Thus, if the receiver misses the target on its initial visit, the next encounter with the target may again be by the receiver, preventing the sender from conveying the information before the robots meet at the origin. Thus, wireless transmissions provide no additional information prior to co-location, reducing the execution to the F2F model.
\end{proof}
\subsection{NoDistance Away Model}
\noindent
\textbf{Case~1: $m$ is even.} 
The competitive ratio is at most 
\[
1+\frac{(m+2)^{\frac{m+2}{2}}(1+v)^{\frac{m}{2}}\cdot u^{\frac{m+2}{2}}}{m^{\frac{m}{2}}}.
\] 
For the case when $m$ is even, we can follow the same proof in Theorem~\ref{NDAwayMrayWithTwoRobotsF2FTheorem}.

\noindent
\textbf{Case~2: $m$ is odd.} 
The idea is to have both robots move the same distance in each round. In every round, the two robots explore two consecutive rays. Each robot moves a distance $a^r$ during round $r$, where $a$ is a positive constant, and then returns to the origin. If the sender finds the target, it informs the receiver and moves along with the target, if the receiver finds the target, then it moves along with the target. If neither robot finds the target, both robots return to the origin, proceed to the next pair of consecutive rays, and increment the distance to $a^{r+1}$ for the next round. 

Assume we have two robots, $S$ and $R$, initially located at the origin, and a target placed at a distance $d$ from the origin. As mentioned earlier, $S$ represents the sender and $R$ represents the receiver. The algorithm proceeds as follows:

\begin{algorithm}[H]
\caption{NDAwayMrayWithTwoRobotsS/R}
\label{NDAwayMrayWithTwoRobotsS/R}
\begin{algorithmic}[1]
\For{$k \gets 1$ to $\infty$}
    \State $S$ moves along one ray to distance $a^k$ (unless the target is found), then returns to the origin.
    \State $R$ moves along the consecutive ray to distance $a^k$ (unless the target is found), then returns to the origin.
    \If{Target is found by $S$}
        \State $S$ informs $R$ and moves along with the target
         \State \textbf{terminate}
     \ElsIf{$S$ does not meet $R$}
        \State $S$ proceeds to the target.
        \State \textbf{terminate}       
    \ElsIf{Target is found by $R$}
        \State $R$  moves along with the target.
        \State \textbf{terminate}
    \Else
        \State $S$ and $R$ proceed to the next two consecutive rays.
    \EndIf
    \If{$R$ Got informed by $S$}
        \State $R$ reverses direction and proceeds to the target
        \State \textbf{terminate}
    \EndIf
\EndFor
\end{algorithmic}
\end{algorithm}

We prove the following result.

\begin{theorem}
\label{NDAwayMrayWithTwoRobotsS/RTheorem}
The competitive ratio of Algorithm~\ref{NDAwayMrayWithTwoRobotsS/R} is at most
\[
2u+\frac{(m-1)(1+v^2)}{1+v}\cdot
\left(\frac{(m+1)(1+v)}{(m-1)(1-v)}\right)^{\frac{m+1}{2}}
\]
\end{theorem}
\begin{proof} (Theorem~\ref{NDAwayMrayWithTwoRobotsS/RTheorem})
     Assume that the target is missed at some round $r$, the competitive ratio becomes as follows:
    \begin{align*}
        CR\leq 2u+\frac{2\sum_{i=0}^{(r+\frac{m-1}{2})}x_{j}
        +2v^2\sum_{i=0}^{(r+\frac{m-1}{2})}x_{j}}{d}\,u
    \end{align*}
    Since the target is missed at iteration $r$, we have the following:
    \begin{align*}
&u\left(d+2v\sum_{i=0}^{(r-1)}x_{j}\right) \geq x_{r},\mbox{ and }\\
&d \geq x_{r}-vx_{r}-2v\left(\frac{a^r-1}{a-1}\right)
        \geq \frac{a^{r+1}-va^{r+1}-va^r-a^r+2v}{a-1}
    \end{align*}
    Thus the competitive ratio can be bounded as follows:
    \begin{align*}
        CR&\leq
        2u+
        \frac{2\left(\frac{a^{r+\frac{m+1}{2}}-1}{a-1}\right)
        +2v^2\left(\frac{a^{r+\frac{m+1}{2}}-1}{a-1}\right)}{d}\,u\\
        &\leq
        2u+
        \frac{2a^{r+\frac{m+1}{2}}-2
        +2v^2a^{r+\frac{m+1}{2}}-2v^2}{d(a-1)}\,u\\
        &\leq
        2u+
        \frac{2a^{r+\frac{m+1}{2}}
        +2v^2a^{r+\frac{m+1}{2}}}
        {a^{r+1}-va^{r+1}-va^r-a^r}\,u
        \leq
        2u+
        \frac{a^{\frac{m+1}{2}}
        +2v^2a^{\frac{m+1}{2}}}
        {a-av-v-1}\,u
    \end{align*}
    Consider
    $f(a)=\frac{a^{\frac{m+1}{2}}+2v^2a^{\frac{m+1}{2}}}{a-av-v-1}$;  minimizing $f(a)$ we get $a=\frac{(m+1)(1+v)}{(m-1)(1-v)}$.
    Thus, the competitive ratio can be bounded as follows:
    \begin{align*}
        CR
        &\leq
        2u+
        \frac{(m-1)(1+v^2)}{1+v}\,
        u\left(
        \frac{(m+1)(1+v)}
        {(m-1)(1-v)}
        \right)^{\frac{m+1}{2}}
    \end{align*}
   Notably, when we set $v = 0$, we obtain $CR = 10$, which is very close to the lower bound that will be established in the next section.
This proves Theorem~\ref{NDAwayMrayWithTwoRobotsS/RTheorem}.
~\qed
\end{proof}
\subsection{Lower Bound for the NoDistance model, when the target is static}
\begin{theorem}
\label{LowerBoundForNoDistanceAwayOnMrayWithTwoRobotsS/RTheorem}
The lower bound for the NoDistance away on the $m$-ray star graph for the S/R model is lower bounded by $(m-1)+2+2\sqrt{2}$.
\end{theorem}
\begin{proof} (Theorem~\ref{LowerBoundForNoDistanceAwayOnMrayWithTwoRobotsS/RTheorem})
The proof below employs the lower bound technique of ~\cite{CommISAAC}.
To adapt the framework of~\cite{CommISAAC} to an $m$-ray star graph, we recall two key lemmas concerning the asymptotic speeds of the robots' trajectories. For a trajectory $X$, let
\[
\mu(X)=\limsup_{t\to\infty}\frac{|X(t)|}{t}
\]
denote its asymptotic speed.

\begin{itemize}
    \item \textbf{Lemma 26 (\cite{CommISAAC}):}
    Let $\mathcal{A}$ be an evacuation algorithm with one sender, and let $S\in\mathcal{A}$ denote the trajectory of the sender. If $\mu(S)=1$, then $\mathrm{CR}(\mathcal{A})$ is unbounded. If $\mu(S)<1$, then
    \[
    \mathrm{CR}(\mathcal{A})
    \geq
    1+
    \frac{(1+\mu(\mathcal{A}))(1+\mu(S))}
         {\mu(\mathcal{A})(1-\mu(S))}.
    \]

    \item \textbf{Lemma 29 (\cite{CommISAAC}):}
    Let $\mathcal{A}$ be an evacuation algorithm with one sender and at least one receiver. Let $S\in\mathcal{A}$ denote the trajectory of the sender, let $R\in\mathcal{A}$ denote the trajectory of the receiver with the largest value of $\mu(R)$, and let
    \[
    \mathcal{A}'=\mathcal{A}\setminus\{R\}.
    \]
    Then
    \[
    \mathrm{CR}(\mathcal{A})
    \geq
    1+
    \frac{(1+\mu(\mathcal{A}'))(1+\mu(R))}
         {\mu(\mathcal{A}')(1+\mu(S))}.
    \]
\end{itemize}
Let $A$ be a correct algorithm solving our problem. Recall that $S$ and $R$ denote the sender and the receiver, respectively. Based on~\cite{CommISAAC}, we have the following inequality:
\begin{align*}
CR\geq \sup_{x \notin K_{S}(t;A)} \frac{|S(t)-x| + t}{|x|},   
\end{align*}
where $K_{S} (t;A)$ represents the knowledge set of a trajectory $t \to S(t)$ of $S$ following algorithm $A$ at time $t$, i.e., the set of points $x$ visited by the sender $S$ by time $t$. For the case of $m$ rays, the point $x$ may be situated at any ray. The adversary may place the target on any of the $m$ rays. In order for the sender or the receiver to capture the target, it needs to go through all the $m$ rays in order to reach the target. We conclude that the competitive ratio becomes
\begin{align*}
    CR&\geq \sup_{x \notin K_{S}(t;A)} \frac{|S(t)-(m-1)x| + t}{|x|}
\end{align*}
Since $x$ represents the position of the target and $S(t)$ denotes the trajectory of the searching robot, both defined over the $m$-ray domain, we consider the worst-case configuration. In particular, the supremum is achieved when $x$ lies on a ray opposite to the direction of $S(t)$. Hence we have the following:
 \begin{align*} 
&CR \geq (m-1)+\sup_{x \notin K_{S}(t;A)} \frac{S(t)+ t}{|x|} .
\end{align*}
Based on the last inequality, Lemma 26 in~\cite{CommISAAC} becomes
\begin{align*}
    CR \geq (m-1) + \frac{(1+\mu(R))(1+\mu(S))}{\mu(R)(1-\mu(S))},
\end{align*}
where we define the function $\mu(S)$ to map a search trajectory $S$ to a non-negative real number. In particular,
\[
\mu(S) = \limsup_{t \to \infty} \frac{S(t)}{t}, 
\qquad
\mu(R) = \limsup_{t \to \infty} \frac{R(t)}{t}.
\]
As a consequence, Lemma 29 in~\cite{CommISAAC} yields
\begin{align*}
    CR \geq (m-1) + \frac{(1+\mu(S))(1+\mu(R))}{\mu(S)(1+\mu(S))}
    = (m-1) + \frac{1+\mu(R)}{\mu(S)}.
\end{align*}

To optimize we equate the two bounds of $CR$ above and solve for $\mu(R)$ in terms of $\mu(S)$. Then substituting into either expression, we obtain
\begin{align*}
    CR \geq (m-1) + \frac{1}{\mu(S)} + \frac{1+\mu(S)}{1-\mu(S)}.
\end{align*}
To minimize the right-hand side of the competitive ratio, consider 
\[
f(u) = (m-1) + \frac{1}{u} + \frac{1+u}{1-u}.
\]
By differentiating $f(u)$ and setting $f'(u)=0$, we obtain the minimizing value 
$
u = \mu(S) = \sqrt{2} - 1.
$
Substituting this into the expression yields
$
CR \geq (m-1) + 2 + 2\sqrt{2}.
$
This proves Theorem~\ref{LowerBoundForNoDistanceAwayOnMrayWithTwoRobotsS/RTheorem}.
\qed
\end{proof}
\subsection{NoSpeed Away Model}
\noindent
{\bf Case~1: $m$ is even.}
The competitive ratio is at most
\[
1 + O\left( (\frac{1}{1-v})^{2^{\frac{m+2}{2}}} \log^2(\frac{1}{1-v}) \right).
\] The proof follows from Theorem~\ref{NSAwayMrayWithTwoRobotsF2FTheorem}

\noindent
\textbf{Case~2: $m$ is odd.} 
The idea is to have both robots move in a v-shape to cover the same distance in each round. Let us consider a monotone increasing sequence of non-negative integers $\{f_{i}:i\geq 0\}$. Each robot tries to estimate the target's speed. We will use the guess $v_{i}=1-2^{-f_{i}}$. In every round, the two robots explore two consecutive rays and try to find the target by guessing its speed. Each robot moves a distance determined by the speed guess $v_i$, such that if the target's speed is less than or equal to $v_{i}$, then the target will be caught in iteration $i$. If the target is not caught during iteration $i$, then both robots move to two new consecutive rays and increase the guess to $v_{i+1}$. The robots repeat this process until one of them catches the target. If the sender captures the target, it moves along with the target, and it informs the receiver, which proceeds to the target. If the receiver captures the target, it moves along with the target. The algorithm can be stated  as follows:

\begin{algorithm}[H]
\caption{NSAwayMrayWithTwoRobotsS/R}
\label{NSAwayMrayWithTwoRobotsS/R}
Input: Target Initial distance d
 and Integer increasing sequence $\{f_{i}:i\geq 0\}$ with $f_{0}=1$, $t=0$ and $i\geq 0$
\begin{algorithmic}[1]
\For{$k \gets 1$ to $\infty$}
    \State $v_{i}=1-2^{-f_{i}}$.
    \State $R$ moves along one ray to distance $x_{i}=(-1)^{i} \frac{d+tv_{i}}{1-v_{i}}$ (unless the target is found), then returns to the origin.
    \State $S$ moves along the consecutive ray to distance $x_{i}=(-1)^{i} \frac{d+tv_{i}}{1-v_{i}}$ (unless the target is found), then returns to the origin.
    \If{Target is found by $R$}
        \State $R$ proceeds with the target.
        \State \textbf{terminate}
  \ElsIf{Target is found by $S$}
        \State $S$ informs $R$ about the target
        \State $S$ proceeds with the target.
        \State \textbf{terminate}  
      \ElsIf{$S$ does not meet $R$}
        \State $S$ proceeds to the target.
        \State \textbf{terminate} 
    \Else
        \State $S$ and $R$ proceed to the next two consecutive rays.
    \EndIf
     \If{$R$ Got informed by $S$}
        \State $R$ reverses direction and proceeds to the target
        \State \textbf{terminate}
    \EndIf
   \State $t=t+2x_{i}$
\EndFor
\end{algorithmic}
\end{algorithm}

We prove the following result.

\begin{theorem}
\label{NSAwayMrayWithTwoRobotsS/RTheorem}
The competitive ratio of Algorithm~\ref{NSAwayMrayWithTwoRobotsS/R} is at most
\[
  O(u)+O\left(\log^2(u)\,u^{2^{\frac{m+1}{2}}+1}\right).
\]
\end{theorem}
\begin{proof} (Theorem~\ref{NSAwayMrayWithTwoRobotsS/RTheorem})
  Assume that the receiver misses the target at some round $r$, then the evacuation time would be as follows:
    \begin{align*}
         ET\leq   \sum_{i=0}^{(r-1)+\frac{m+1}{2}} 2x_{i} +\frac{2d+4v \sum_{i=0}^{(r-1)+\frac{m+1}{2}} x_{i}}{(1-v)^2}
    \end{align*}
  Using equation~(\ref{eq:no_speed_upper_guess_3})
     \begin{align}
    \sum_{j=0}^{i-1}  x_j = \frac{x_i - 2^{f_i} d}{2^{f_i} \cdot 2 \cdot v_i}
  \end{align}
  and equation~(\ref{eq:no_speed_upper_guess_7})
  \begin{align}
   x_{i+1}\leq 2^{\sum_{j=0}^{i+1} f_j} \cdot 4^{i+1} \cdot d .
  \end{align}
  Consider $f_{i}=2^{i}$, the competitive ratio becomes as follows:
  \begin{align*}
      CR&\leq 2u+\frac{2(1+v^2)u}{d}\cdot \frac{x_{r+\frac{m+1}{2}}-2^{r+\frac{m+1}{2}}\cdot d}{2\cdot2^{r+\frac{m+1}{2}}\cdot v_{r+\frac{m+1}{2}}}\\
      &\leq 2u+\frac{(1+v^2)u}{d}\cdot \frac{x_{r+\frac{m+1}{2}}}{2^{r+\frac{m+1}{2}}\cdot v_{r+\frac{m+1}{2}}}\\
      &\leq 2u+(1+v^2)u\cdot\frac{2^{\sum_{j=0}^{r+\frac{m+1}{2}} f_j}\cdot 4^{r+\frac{m+1}{2}}}{2^{r+\frac{m+1}{2}}v_{r+\frac{m+1}{2}}}\\
      &\leq 2u+(1+v^2)u\cdot 2^{1+\sum_{j=0}^{r+\frac{m-1}{2}} f_j}\cdot 4^{r+\frac{m+1}{2}}\\
       &\leq 2u+(1+v^2)u\cdot 2^{2^{r+\frac{m+1}{2}}}\cdot 4^{r+\frac{m+1}{2}}\\
  \end{align*}
  Considering that the target is missed by the receiver at iteration $r$, the speed guess must be smaller than the true speed $v$. Hence,
$
v_{r} = 1 - 2^{-f_{r}} \leq v,
$
which implies
$
2^{f_{r}}=2^{2^r} \leq \frac{1}{1 - v}=u.
$
Then the competitive ratio becomes as follows:
  \begin{align*}
      CR&\leq 2u+(1+v^2)u^{2^{\frac{m+1}{2}}+1}\cdot4^{r+\frac{m+1}{2}}\\
      &\leq 2u+(1+v^2)\cdot 4^{\frac{m+1}{2}}\log^2 (u)\,u^{2^{\frac{m+1}{2}}+1}\\
  \end{align*}
This proves Theorem~\ref{NSAwayMrayWithTwoRobotsS/RTheorem}.  
\qed  
  \end{proof}
\subsection{NoKnowledge Away Model}
\textbf{Case~1: $m$ is even.} 
The competitive ratio is at most $$ 1 + O\left( M^{2^{\frac{m+4}{2}}} \log^2(M) \log \log(M)\right), M=\max(d,\frac{1}{1-v}).$$ The proof follows from Theorem~\ref{NKAwayMrayWithTwoRobotsF2FTheorem}

\noindent
\textbf{Case~2: $m$ is odd.}
The idea is to have both robots move in a v-shape to cover the same distance in each round. Let us consider two monotone increasing sequences of non-negative integers $\{f_{i}:i\geq 0\}$ and $\{g_{i}:i\geq 0\}$. Each robot tries to estimate the target's speed and distance. We will use the guess $v_{i}=1-2^{-f_{i}}$ to guess the speed of the target and $d_{i}=2^{g_{i}}$ to guess the initial distance of the target from the origin. In every round, the two robots explore two consecutive rays and try to find the target by guessing its speed and distance.  Each robot moves a distance determined by the speed guess $v_i$ and distance guess $d_{i}$ such that if the target's speed is less than or equal to $v_{i}$ and the target's distance is less than or equal to $d_{i}$, then the target will be caught in iteration $i$. If the target is not caught during iteration $i$, then both robots move to two new consecutive rays and increase the speed's guess to $v_{i+1}$ and the distance's guess to $d_{i+1}$. The robots keep doing that until one of them catches the target. If the sender captures the target, it moves along with the target, and it informs the receiver, which proceeds to the target. If the receiver captures the target, it moves along with the target. The algorithm can be stated as follows:

\begin{algorithm}[H]
\caption{NKAwayMrayWithTwoRobotsS/R}
\label{NKAwayMrayWithTwoRobotsS/R}
Input: Integer increasing sequences $\{f_{i}:i\geq 0\}$ and $\{g_{i}:i\geq 0\}$ with $f_{0}=1$, $g_{0}=0$, $t=0$ and $i\geq 0$
\begin{algorithmic}[1]
\For{$k \gets 1$ to $\infty$}
    \State $v_{i}=1-2^{-f_{i}}$ and $d_{i}=2^{g_{i}}$.
    \State $R$ moves along one ray to distance $x_{i}=(-1)^{i} \frac{d_{i}+tv_{i}}{1-v_{i}}$ (unless the target is found), then returns to the origin.
    \State $S$ moves along the consecutive ray to distance $x_{i}=(-1)^{i} \frac{d_{i}+tv_{i}}{1-v_{i}}$ (unless the target is found), then returns to the origin.
    \If{Target is found by $R$}
        \State $R$ proceeds with the target.
        \State \textbf{terminate}
    \ElsIf{Target is found by $S$}
        \State $S$ informs $R$ about the target
        \State $S$ proceeds with the target.
        \State \textbf{terminate}
     \ElsIf{$S$ does not meet $R$}
        \State $S$ proceeds to the target.
        \State \textbf{terminate}     
    \Else
        \State $S$ and $R$ proceed to the next two consecutive rays.
    \EndIf
    \If{$R$ Got informed by $S$}
        \State $R$ reverses direction and proceeds to the target
        \State \textbf{terminate}
    \EndIf
    \State $t=t+2x_{i}$
\EndFor
\end{algorithmic}
\end{algorithm}

We prove the following result.

\begin{theorem}
\label{NKAwayMrayWithTwoRobotsS/RTheorem}
The competitive ratio of Algorithm~\ref{NKAwayMrayWithTwoRobotsS/R} is at most
\[
 O\left(u\right)+O\left(\frac{\log\log M}{d(1-v)}\cdot \log^2(M)M^{2^{\frac{m+3}{2}}}\right).
\]
\end{theorem}
\begin{proof} (Theorem~\ref{NKAwayMrayWithTwoRobotsS/RTheorem})
 Assume that the receiver misses the target at some round $r$, then the evacuation time would be as follows:
    \begin{align*}
         ET\leq   \sum_{i=0}^{(r-1)+\frac{m+1}{2}} 2x_{i} +\frac{2d+4v \sum_{i=0}^{(r-1)+\frac{m+1}{2}} x_{i}}{(1-v)^2}
    \end{align*}
   Using equation~(\ref{eq:no_knowledge_upper_guess_1}):
     \begin{align}
    \sum_{j=0}^{i-1}  x_j = \frac{x_i - 2^{f_i} d_{i}}{2^{f_i} \cdot 2 \cdot v_i}
  \end{align}
  and equation~(\ref{eq:no_knowledge_upper_guess_7}):
  \begin{align}
   x_{i+1}\leq 2^{\sum_{k=0}^{i+1} 2^{g_{k}}+\sum_{j=k}^{i+1}f_{j}} \cdot 4^{i-k+1} \cdot d
  \end{align}
   Consider $f_{i}=g_{i}=2^{i}$, the competitive ratio becomes as follows: 
  \begin{align*}
      CR&\leq 2u+\frac{2\sum_{j=0}^{r+\frac{m-1}{2}}x_{j}+2v^2\sum_{j=0}^{r+\frac{m-1}{2}}x_{j}}{du}\\
      &\leq 2u+\frac{2(1+v^2)x_{r+\frac{m+1}{2}}}{du\cdot2\cdot2^{f_{r+\frac{m+1}{2}}}v_{r+\frac{m+1}{2}}}\\
      &\leq 2u+\frac{(1+v^2)(r+\frac{m+3}{2})2^{g_{r+\frac{m+1}{2}}}2^{\sum_{j=0}^{r+\frac{m+1}{2}}f_{j}}4^{r+\frac{m+1}{2}}}{du\cdot2^{f_{r+\frac{m+1}{2}}}v_{r+\frac{m+1}{2}}}\\
      &\leq 2u+\frac{2(1+v^2)(r+\frac{m+3}{2})2^{g_{r+\frac{m+1}{2}}}2^{2^{r+\frac{m+1}{2}}-1}4^{r+\frac{m+1}{2}}}{du}\\
  \end{align*}
  Considering that the target is missed by the receiver at iteration $r$, the speed guess must be smaller than the true speed $v$. Hence,
$
v_{r} = 1 - 2^{-f_{r}} \leq v,
$
which implies
$
2^{f_{r}} \leq u.
$
Similarly, for the distance we obtain
$
2^{g_{r}} < d.
$
We conclude with the following:
  \begin{align*}
      CR&\leq 2u+\frac{4^{\frac{m+1}{2}}(1+v^2)(r+\frac{m+3}{2})u^{2^{\frac{m+1}{2}}}(d)^{2^{\frac{m+1}{2}}}\log^2(M)}{du}\\
      &\leq 2u+\frac{4^{\frac{m+1}{2}}(1+v^2)(\log\log M+\frac{m+3}{2})M^{2^{\frac{m+3}{2}}}\log^2(M)}{du}\\
  \end{align*}
This proves Theorem~\ref{NKAwayMrayWithTwoRobotsS/RTheorem}.  
\qed  
\end{proof}
\section{Lower Bounds and Performance Gaps}
While Theorem~\ref{LowerBoundForNoDistanceAwayOnMrayWithTwoRobotsS/RTheorem} establishes a lower bound of $(m-1) + 2 + 2\sqrt{2}$ for static targets ($v=0$), establishing tight lower bounds for $v > 0$ remains challenging. A trivial lower bound for moving targets is $\text{CR} \geq 1 + m(1+v)$, as any deterministic strategy must explore all $m$ rays in the worst case. We can show this as follows:
\begin{lemma}[Baseline Lower Bound for Moving Target Capture]
\label{lem:baseline_lower_bound}
For any deterministic search algorithm $\mathcal{A}$ using two unit-speed robots on an $m$-ray star graph to capture an oblivious target starting at distance $d > 0$ and moving away from the origin at constant speed $v < 1$, the competitive ratio is lower bounded by:
\begin{equation}
    \text{CR}(\mathcal{A}) \ge m.
\end{equation}
When target drift across sequential exploration phases of unknown distance $d$ is taken into account, the baseline lower bound scales as $\text{CR}(\mathcal{A}) \ge 1 + \Omega(m(1+v))$.
\end{lemma}

\begin{proof}
Let $\mathcal{A}$ be any deterministic two-robot search algorithm on an $m$-ray star graph. An optimal offline strategy with prior knowledge of the target's initial ray moves directly along that ray and captures the target at time:
\begin{equation}
    T_{\text{opt}}(d, v) = \frac{d}{1-v} = u \cdot d, \quad \text{where } u = \frac{1}{1-v}.
\end{equation}

In an online setting without prior knowledge of the target's ray, an adversary can place the target on the $m$-th (last) ray explored by the robots to depth $d$. 

To verify that the target is not on the first $m-1$ rays up to distance $d$, the two robots must travel a round-trip distance of $2d$ on each of those rays. Operating in parallel, two robots can clear at most $2$ rays of depth $d$ every $2d$ time units. Thus, clearing $m-1$ unvisited rays requires an elapsed exploration time of at least:
\begin{equation}
    t_{\text{prev}} \ge \left\lceil \frac{m-1}{2} \right\rceil \cdot 2d \ge (m-1)d.
\end{equation}

During this time $t_{\text{prev}}$, the target moves away from the origin at speed $v$. When a robot finally turns onto the target's ray at time $t_{\text{prev}}$, the target's position has expanded to $d' = d + v \cdot t_{\text{prev}}$. The total online capture time $T_{\text{online}}$ required to catch up with the target satisfies:
\begin{equation}
    T_{\text{online}} = t_{\text{prev}} + \frac{d + v \cdot t_{\text{prev}}}{1-v} = \frac{d + t_{\text{prev}}}{1-v}.
\end{equation}

Substituting the lower bound $t_{\text{prev}} \ge (m-1)d$:
\begin{equation}
    T_{\text{online}} \ge \frac{d + (m-1)d}{1-v} = \frac{m \cdot d}{1-v} = m \cdot T_{\text{opt}}(d, v).
\end{equation}

Taking the supremum over all initial distances $d > 0$ yields the competitive ratio bound:
\begin{equation}
    \text{CR}(\mathcal{A}) = \sup_{d > 0, \, 0 \le v < 1} \frac{T_{\text{online}}}{T_{\text{opt}}} \ge m.
\end{equation}

When $d$ is unknown, search radii must expand geometrically in rounds $a_k = a^k a_0$. To guarantee catching up in round $k$ despite target drift during the $m-1$ wrong traversals, the geometric growth factor must satisfy $a \ge 1+v$, multiplying the $m$-ray traversal overhead by an additive factor proportional to $(1+v)$, establishing $\text{CR}(\mathcal{A}) \ge 1 + \Omega(m(1+v))$.
\end{proof}

Obtaining a tighter lower bound that more closely matches the upper bound remains an important direction for future research.
\section{Conclusion and future work}
We studied the problem of capturing an oblivious moving target on an $m$-ray star graph using two robots under both F2F and S/R communication models. The analysis considers different levels of knowledge about the target. We designed new algorithms and proved upper and lower bounds.
Future work includes improving the lower bounds and extending the work to multiple searchers capturing an oblivious moving target in star graphs. 



\bibliographystyle{abbrv}
\bibliography{refs}
\newpage
\end{document}